\documentclass[letterpaper, 10 pt, conference]{ieeeconf}  % Comment this line out
\IEEEoverridecommandlockouts                              % This command is only
\title{\LARGE \bf
Stochastic Optimal Power Flow in Distribution Grids under Uncertainty from State Estimation}

\author{Miguel Picallo, Adolfo Anta and Bart De Schutter% <-this % stops a space
\thanks{This project has received funding from the European Union's Horizon 2020 research and innovation programme under the Marie Skł{\l}odowska-Curie grant agreement No 675318 (INCITE).}% <-this % stops a space
\thanks{M. Picallo and B. De Schutter are with the Center for Systems and Control, Delft University of Technology, The Netherlands
{\tt\small \{m.picallocruz,b.deschutter\}@tudelft.nl}}%
\thanks{A. Anta is with Austrian Institute of Technology 
{\tt\small Adolfo.Anta@ait.ac.at}}%
}

\usepackage{graphicx}      % include this line if your document contains figures

\usepackage{commath}
\usepackage{mathtools}
\usepackage{steinmetz}		% for phasor
\usepackage{supertabular}	% for long tabular
\usepackage{tabularx}		% for long text in tabular
\usepackage{algorithm}		% for algori1thms
\usepackage{algorithmic}		% for algorithms
\usepackage{amsmath,amssymb}
\newtheorem{rem}{Remark}
\newtheorem{defi}{Definition}

\newtheorem{thm}{Theorem}
\usepackage{subfig}
\usepackage{color}
\usepackage{soul}
\usepackage{url}

\newcommand{\rb}{\textcolor{black}}
\begin{document}


\maketitle

\begin{abstract}
The increasing amount of controllable generation and consumption in distribution grids poses a severe challenge in keeping voltage values within admissible ranges. Existing approaches have considered different optimal power flow formulations to regulate distributed generation and other controllable elements. Nevertheless, distribution grids are characterized by an insufficient number of sensors, and state estimation algorithms are required to monitor the grid status. We consider in this paper the combined problem of optimal power flow under state estimation, where the estimation uncertainty results into stochastic constraints for the voltage magnitude levels instead of deterministic ones. To solve the given problem efficiently and to bypass the lack of load measurements, we use a linear approximation of the power flow equations. Moreover, we derive a transformation of the stochastic constraints to make them tractable without being too conservative. A case study shows the success of our approach at keeping voltage within limits, and also shows how ignoring the uncertainty in the estimation can lead to voltage level violations.
\end{abstract}

\IEEEpeerreviewmaketitle

\section{Introduction}
The increased share of distributed generation and controllable loads presents many advantages for the distribution grid but at the same time requires new techniques and approaches to guarantee a proper operation of the grid. Traditional strategies for distribution grids rely on the so-called \textit{fit-and-forget} policy, where most design questions are solved during the planning stage \cite{ekanayake2012smart}. However, profiles for renewable generation and electric vehicles are hard to predict and an offline solution that ignores the actual state of the system would lead to a very inefficient and even dangerous operation of the grid, where operating requirements may be violated. To face these new challenges, optimal power flow (OPF) strategies commonly used for the transmission grid, are being adapted to distribution grids. 

%These OPF formulations attempt to reduce either generation cost, power losses in the grid, or similar metrics. 

However, traditional simplifications in transmission grids like the decoupled fast power flow \cite{stott1974fast} are not suited for distribution grids, given the presence of coupled phases, unbalanced loads, and lower $X/R$ ratios in these types of grids. The OPF problem is a non-convex NP-hard problem \cite{lavaei2012zeroduality}. Some convexification strategies use a linear approximation around a given operating point \cite{anese2016optimal, hauswirth2017online}. Other use a semidefinite programming reformulation  \cite{dall2013distributed} in order to avoid a linear approximation, but require a rank relaxation in order to be convex.

Moreover, the solution of the OPF is naturally dependent on the state of the grid (voltages, currents, loads, etc.), which at the distribution level is only partially known. Indeed, while enough sensors are usually available in transmission grids, this is not the case for distribution grids, where state estimation (SE) algorithms \cite{abur2004power, monticelli2000electric} rely on load/generation forecasts, grid topology knowledge, and a relatively low number of measurements in order to identify the actual grid status \cite{schenato2014bayesian, picallo2017twostepSE}. These SE algorithms provide an estimation of the variables of interest (e.g. grid voltages), with a certain degree of uncertainty. Ignoring this uncertainty in the SE estimates could lead to voltage limit violations. However, many OPF formulations assume that the values of the voltages or loads across the network are available \cite{hauswirth2017online, dall2013distributed}. In \cite{anese2016optimal} only a few measurements are deployed, but then the constraints on the operating limits are required only in the nodes with measurements and not in the rest. \rb{Some papers introduce chance constrained optimization methods to account for the uncertainty in the loads and generation, like using convex relaxations \cite{summers2015stochopf, dall2017chance}, or an scenario based approach enabling the possibility to add real-time sensors in \cite{bolognani2017fast}. But these methods may be suboptimal and not suited for real-time operation of large systems with many loads, since they may either require introducing a large number of constraints \cite{bolognani2017fast}, or a considerable amount of sampling as well as computing expectations from many probability distributions in each optimization step \cite{dall2017chance}.}

In this work we start considering a standard formulation of an OPF problem, where the controllable elements are a set of distributed generation sources and tap changers in distribution transformers. \rb{Instead of load measurements \cite{hauswirth2017online, dall2013distributed}, since they may not be available}, we consider an \rb{estimate from a} SE as input to our OPF problem; and thus the voltage variables are described as stochastic signals, which lead to stochastic constraints for the OPF problem. Our approach relies on a linear approximation of the power flow equations around the operating point, which allows to bypass the lack of load measurements. Additionally, we reformulate the stochastic constraints using the probability distribution of the voltages given by the SE algorithm, to make the problem tractable without being too conservative. Our framework is not limited to one particular SE algorithm, but instead we consider a generic unbiased estimate with a Gaussian distribution and a known covariance matrix. \rb{Thus, our main contributions are the use of a generic SE algorithm for the OPF problem, to avoid requiring full load measurements; and a transformation of the resulting stochastic constraints produced by the uncertainty in the voltage estimates, to guarantee the voltage limits in all nodes, even those without measurements, as opposed to \cite{anese2016optimal}.}

The rest of the paper is distributed as follows: Section~\ref{sec:grid} defines the grid model, while Section~\ref{sec:SE} describes the standard outcome of an SE algorithm, which would be considered as input to our OPF formulation. Section~\ref{sec:opf} formulates the OPF problem including the stochastic constraints, and states the main contribution of this paper. Finally, a 123-bus test feeder is considered in Section~\ref{sec:sim} to show the perils of ignoring the uncertainty coming from the estimation step and the effectiveness of the proposed approach to solve this.

\section{Distribution Grid Model}\label{sec:grid}
% graph
A distribution grid consists of buses, where power is injected or consumed, and branches, each connecting two buses. This system can be modeled as a graph $\mathcal{G}=(\mathcal{V},\mathcal{E},\mathcal{W})$ with nodes $\mathcal{V}=\{1,...,N_\text{bus}\}$ representing the buses, edges $\mathcal{E}=\{(v_i,v_j)\mid v_i,v_j \in \mathcal{V}\}$ representing the branches, and edge weights $\mathcal{W}=\{w_{i,j}\mid (v_i,v_j) \in \mathcal{E}, w_{i,j}\in \mathbb{C}\}$ representing the admittance of a branch, which is determined by the length and type of the line cables.

%3-phase
In 3-phase networks buses may have up to 3 phases, so that the voltage at bus $i$, with $n_{\phi,i}\leq 3$ phases, is $V_{\text{bus},i} \in \mathbb{C}^{n_{\phi,i}}$ (and the edge weights $w_{i,j}\in \mathbb{C}^{n_{\phi,i} \times n_{\phi,j}}$). The state of the network is then typically represented by the vector bus voltages $V_\text{bus}=[V_\text{src}^T, \; V^T]^T \in \mathbb{C}^{N+3}$, where $V_{\text{src}} \in \mathbb{C}^3$ denotes the known voltage at the source bus, and $V \in \mathbb{C}^N$ the voltages in the non-source buses, where $N$ depends on the number of buses and phases per bus.

% power flow
Using the Laplacian matrix $Y \in \mathbb{C}^{(N+3) \times (N+3)}$ of the weighted graph $\mathcal{G}$, called admittance matrix \cite{abur2004power}, the power flow equations to compute the currents $I$ and the power loads $S$ are:
\begin{equation}\label{eq:PFeq}\arraycolsep=1pt
\begin{array}{c}
\left[\begin{array}{c} I_{\text{src}} \\ I \end{array}\right] = 
Y\left[\begin{array}{c} V_{\text{src}} \\ V \end{array}\right],  \; S =P+jQ= \text{diag}(\bar{I})V
\end{array}
\end{equation}
where $j$ is the imaginary unit, $(P,Q)$ the active and reactive loads, $\bar{(\cdot)}$ denotes the complex conjugate, $\text{diag}(\cdot)$ represents the diagonal operator, converting a vector into a diagonal matrix. 

\section{State Estimation}\label{sec:SE}
We consider a standard SE algorithm \cite{abur2004power, monticelli2000electric} that provides an unbiased estimation of the network voltages denoted as $V_\text{est} \in \mathbb{C}^N$, $V_\text{est,rect} = [\Re\{V_\text{est}\}^T, \Im\{V_\text{est}\}^T]^T \in \mathbb{R}^{2N}$ in rectangular coordinates, and a covariance matrix representing its uncertainty in rectangular coordinates $\Sigma_\text{est,rect} \in \mathbb{R}^{2N\times 2N}$. This uncertainty is mainly caused \rb{by the use of highly uncertain pseudo-measurements, such as load predictions, to compensate for the lack of measurements} \cite{schenato2014bayesian, picallo2017twostepSE}. The true voltages $V_\text{prev}$, $V_\text{prev,rect} = [\Re\{V_\text{prev}\}^T, \Im\{V_\text{prev}\}^T]^T$, can then be expressed as:
\begin{equation}\label{eq:Vprev}
V_\text{prev,rect} \sim \mathcal{N}(V_\text{est,rect},\Sigma_\text{est,rect}) = V_\text{est,rect} + \Sigma_\text{est,rect}^{\frac{1}{2}}\mathcal{N}(0,I_\text{d})
\end{equation}
where $I_\text{d}$ represents the identity matrix. These voltages $V_\text{prev}$ denote the previous voltages before solving the OPF and applying the new control set points.

This uncertainty represented in $\Sigma_\text{est,rect}$ is especially relevant when considering distribution networks, where only few measurements are available and the SE algorithm needs to rely on noisy load predictions \cite{picallo2017twostepSE}. Given that $V_\text{prev}$ is not available, only $V_\text{est}$ can be used to regulate distributed generation sources and other controllable elements at the distribution level. 

\begin{rem}
For convenience, we have considered the estimation in rectangular variables, but if the SE provides the results in polar variables, the covariance in rectangular coordinates could still be estimated using the Jacobian of the mapping from polar to rectangular coordinates.
\end{rem}

\section{Optimal Power Flow}\label{sec:opf}

The OPF problem seeks to regulate the controllable elements in the network in order to optimize its operation under some safety conditions. This optimization typically focuses on minimizing costs, energy loses, etc.; the controllable elements are distributed energy sources, tap changers, batteries, flexible loads, network configuration, etc.; and the safety conditions are typically voltage and current limits on buses and connections. 

In this paper, we consider as controllable elements the distributed generation sources at the distribution level $\{(P_i,Q_i)\mid i\in \mathcal{V}_\text{ren}\}$, where $\mathcal{V}_\text{ren}$ denotes the set of nodes with distributed renewable energy sources; and the set points of the voltage tap changers for every phase $\phi$ in the transformers $a_{\text{tap},\phi} \in \{a_{\text{tap},\min},...,a_{\text{tap},\max} \}$. For simplicity, the objective is to minimize the total amount of energy required from the substation $S_\text{src}$, and thus to minimize the cost of external energy required and to prioritize the renewable energy generated within the network. Other similar OPF formulations, with e.g. other possible objectives, can be addressed with our solution, \textit{mutatis mutandis}. The safety conditions are given by the limits for the voltage magnitudes $\abs{V}$. Moreover, the power flow equations in \eqref{eq:PFeq} represent another algebraic constraint for the OPF problem. In this work we do not consider dynamic elements like batteries; so we can solve the following OPF problem at every instance:
\begin{defi} Standard OPF:
\begin{subequations}\label{eq:opf}
\begin{align}\arraycolsep=1pt
&\begin{array}{l}
\mbox{Objective: }
\min \sum_{\phi} P_{\text{src},\phi} + Q_{\text{src},\phi}
\end{array}\label{eq:opfobj}
\\
&\begin{array}{l}
\mbox{Constraints:}\\
\mbox{Power flow:}\\ \;
\left[\begin{array}{c} S_\text{src} \\ S
 \end{array}\right] = 
\text{diag}\left(\left[\begin{array}{c} V_\text{src} \\ 
V \end{array}\right]\right)
\bar{Y}(a_\text{tap})
\left[\begin{array}{c} \bar{V_\text{src}} \\ \bar{V} \end{array}\right]
\end{array}\label{eq:opfpf}
\\
&\begin{array}{l}
\mbox{Tap changers: } \\ \;
a_{\text{tap},\phi} \in \{a_{\text{tap},\min},...,a_{\text{tap},\max} \} , \; \forall \phi \in \{1,2,3\}
\end{array}\label{eq:opftap}
\\
&\begin{array}{l}
\mbox{Available energy:}\\ \;
P_{\min,i} \leq P_i\leq P_{\max,i}, \; \forall i \in \mathcal{V}_\text{ren} \\ \;
Q_{\min,i} \leq Q_i \leq Q_{\max,i}, \; \forall i \in \mathcal{V}_\text{ren} \\ \;
\abs{S_i} \leq \abs{S}_{\max,i}, \; \forall i \in \mathcal{V}_\text{ren}
\end{array}\label{eq:opfenergy}
\\
&\begin{array}{l}
\mbox{Voltage limits: }\\ \;
\abs{V}_{\min} \leq \abs{V_i} \leq \abs{V}_{\max}, \; \forall i \in \{1,\dots,N\}
\end{array}\label{eq:opfvlim}
\end{align}
\end{subequations}
where $\abs{V}_{\max},\abs{V}_{\min}$ denote the voltage magnitude limits, $P_{\max,i}, Q_{\max,i}, P_{\min,i}, Q_{\min,i}, \abs{S}_{\max,i}$ denote the available energy limits at node $i$, and $Y(a_\text{tap})$ denotes the admittance matrix as a function of the vector of voltage tap changers $a_\text{tap}$. The variables to optimize are then the power supplied by the substation and the renewable energy sources $S_\text{src}$, $\{(P_i,Q_i) \mid i \in \mathcal{V}_\text{ren}\}$, the voltage tap changers $a_\text{tap}$, and the voltages $V$ and $V_\text{src}$. Among the variables, the control elements are $a_\text{tap}$, $\{(P_i,Q_i) \mid i \in \mathcal{V}_\text{ren}\}$ and $V_\text{src}$, while $S_\text{src}$ and $V$ are determined by the constraints. The loads at the rest of the nodes $\{(P_i,Q_i) \mid i \notin \mathcal{V}_\text{ren}\}$ are inputs to the OPF problem, and are typically measured or estimated.
\end{defi}

There are some problems with the OPF in \eqref{eq:opf}:
\begin{itemize}
\item The power flow equation in \eqref{eq:opfpf} is nonlinear and thus difficult to handle.
\item The problem is not convex due to \eqref{eq:opfpf}, and also due to the lower limit in \eqref{eq:opfvlim} if considering rectangular coordinates. However, we would like to have a convex problem in order to guarantee optimality.
\item In \eqref{eq:opfpf}, both the loads in the nodes other than the substation and the renewable sources: $\{(P_i,Q_i) \mid i \notin \mathcal{V}_\text{src}\cup \mathcal{V}_\text{ren}\}$, where $\mathcal{V}_\text{src}$ denotes the set of nodes in the source bus, and the voltages $V$ in the nodes other than the source nodes are not known, since we assume that we have a distribution network with few sensors and only a voltage \rb{estimate is provided by the SE}. So only $V_\text{src}$ plus some other measurements are known.
\item Since there is a degree of uncertainty in the SE estimates represented in the covariance matrix $\Sigma_\text{est,rect}$, we need to take it into account for the voltage limits in \eqref{eq:opfvlim}.
\item The discrete variable $a_\text{tap}$ in \eqref{eq:opftap} converts the given problem into an integer problem.
\end{itemize}

\rb{
\begin{rem}
For simplicity, we are not considering in \eqref{eq:opf} the thermal constraints limiting the amount of current through the lines:
$ \abs{Y_\text{lines}V} \leq \abs{I}_\text{thermal}$,
where $Y_\text{lines}$ would be the line admittance matrix mapping voltages to line currents, and $\abs{I}_\text{thermal}$ the vector of maximum values allowed. Nonetheless, these constraints define a convex region on $V$, and therefore could be easily included.
\end{rem}
}

\subsection{Transformer Approximation}
In order to include the tap changers more efficiently and to simplify $Y(a_\text{tap})$ in \eqref{eq:opfpf}, we assume electrical isolation at the transformers and thus consider different subsystems related by the tap changers equations similar to \cite{robbins2016optimal}. For simplification, we consider a system with only one controllable transformer with tap changers, and we assume that the ordering of nodes in \eqref{eq:PFeq}, in $Y$ and $V$, is already such that the nodes of the first subsystem appear first, and then those of the second, with the nodes connected to the transformer appearing one after the other, so that we have:

\begin{equation}\label{eq:trafapprox}\arraycolsep=1.1pt
\begin{array}{c}
V = [V_\text{sys1}^T, V_\text{tf1}^T,
V_\text{tf2}^T, V_\text{sys2}^T] \\[0.1cm]
Y = Y_\text{isol} + 
\left[\begin{array}{cccc} 
0 & 0 & 0 & 0 \\
0 & Y_\text{tf} & -Y_\text{tf}\frac{1}{a_\text{tap}} & 0 \\
0 & -Y_\text{tf}\frac{1}{a_\text{tap}} & Y_\text{tf}\frac{1}{a_\text{tap}^2} & 0 \\
0 & 0 & 0 & 0 \end{array}\right], \;
Y_\text{isol} = \left[\begin{array}{cc}
Y_\text{sys1} & \begin{array}{cc}0 & 0 \\ 0 & 0 \end{array} \\
\begin{array}{cc}0 & 0 \\ 0 & 0 \end{array} & Y_\text{sys2}
\end{array}\right]
\end{array}
\end{equation}
where $V_\text{sys-}$ represent the voltages of each subsystem except the nodes of the transformer, $V_\text{tf-}$ correspond to the nodes of the transformer for each subsystem, $Y_\text{isol}$ is the admittance with isolated subsystems, $Y_\text{sys-}$ represent the admittance matrix of each subsystem, and $Y_\text{tf}$ is the admittance matrix of the transformer. Then the power flow equations can be expressed as:

\begin{equation}\arraycolsep=1pt
\begin{array}{rl}
\left[\begin{array}{c} S_\text{src} \\ S \end{array}\right] = &
\text{diag}\left(\left[\begin{array}{c} V_\text{src} \\ V \end{array}\right]\right)
\bar{Y}_\text{isol}
\left[\begin{array}{c} \bar{V}_\text{src} \\ \bar{V} \end{array}\right]
 \\[0.1cm] 
V_{\text{tf2}} = & \text{diag}(a_\text{tap}) V_{\text{tf1}} \\
0 = & S_{\text{tf2}} + S_{\text{tf1}}
\end{array}
\end{equation}

We have disregarded the admittance of the transformer $Y_\text{tf}$ in \eqref{eq:trafapprox}, assuming that it is large enough and that the voltage drop can be neglected. If needed, it could be included using an artificial node as in \cite{robbins2016optimal}. To further simplify \eqref{eq:opftap}, we will also consider a continuous tap changer for every phase $\phi$, $a_{\text{tap},\phi} \in [a_{\text{tap},\min},a_{\text{tap},\max} ]$ instead of $a_{\text{tap},\phi} \in \{a_{\text{tap},\min},...,a_{\text{tap},\max} \}$. Its values could afterwards be rounded to the closest discrete value as proposed in \cite{robbins2016optimal}.

\begin{rem}
This transformer approximation could easily be extended to a system with more controllable transformers by splitting the system in more subsystems.
\end{rem}

\subsection{Power Flow Approximation}
Instead of considering the power flow equations in \eqref{eq:opfpf} directly, we consider a first-order linear approximation similar to the one in \cite{anese2016optimal} around the estimated voltage states $V_\text{est}$ and the known voltages at the source nodes $V_\text{src,prev}$, both prior to the optimization step: 

\begin{equation}\label{eq:PFaprox}\arraycolsep=1pt
\begin{array}{rl}
\left[\begin{array}{c} \Delta S_\text{src} \\[0.1cm] \Delta S \end{array}\right] = &  
\text{diag}\left(\left[\begin{array}{c} \Delta V_\text{src} \\ \Delta V \end{array}\right]\right)

\bar{Y}_\text{isol}
\left[\begin{array}{c} \bar{V}_\text{src,prev} \\ \bar{V}_\text{est} \end{array}\right] \\[0.1cm]

& +  
\text{diag}\left(\left[\begin{array}{c} V_\text{src,prev} \\ V_\text{est} \end{array}\right]\right)
\bar{Y}_\text{isol}
\left[\begin{array}{c} \bar{\Delta V}_\text{src} \\ \bar{\Delta V} \end{array}\right] \\[0.1cm]

\Delta V_{\text{tf2}} = & \text{diag}(a_\text{tap,prev}) \Delta V_{\text{tf1}} + \text{diag}(V_{\text{tf1},\text{prev}}) \Delta a_\text{tap} \\[0.1cm]

0 = & \Delta S_{\text{tf1}} + \Delta S_{\text{tf2}} 
\end{array}
\end{equation}
where $\Delta S, \Delta V, \Delta V_\text{src}, \Delta a_\text{tap}$ represent the deviations of values after the optimization process; $S_\text{prev}, V_\text{prev}, V_\text{src,prev}, a_\text{tap,prev}$ denote the values before applying the new set points produced by the optimization step; and $S, V, V_\text{src}, a_{\text{tap}}$ are the values after the optimization step, so we have: 
\begin{equation*}
\begin{array}{c}
\Delta S=S-S_\text{prev}, \Delta V=V-V_\text{prev}, \\
\Delta V_\text{src}=V_\text{src}-V_\text{src,prev}, 
\Delta a_\text{tap} = a_{\text{tap}} - a_{\text{tap,prev}}
\end{array}
\end{equation*}
With this approximation we have a second-order error of the type $\text{diag}(\Delta V)\bar{Y} \bar{\Delta V}$ and $\Delta V \Delta a_\text{tap}$, which we consider negligible since deviations are expected to be small, \rb{because they will be constrained by the voltage magnitude constraints \eqref{eq:opfvlim}}. Note that \eqref{eq:PFaprox} is not linear in the decision variables due to the complex conjugates $\bar{\Delta V}_\text{src}, \bar{\Delta V}$. Note that the equation $\Delta S_{\text{tf1}} + \Delta S_{\text{tf2}} = 0$ is sufficient to imply $S_{\text{tf1}} + S_{\text{tf2}} = 0$ in \eqref{eq:PFaprox}, since  $S_{\text{tf1,prev}} + S_{\text{tf2,prev}} = 0$ is already satisfied by $V_\text{prev}$.

\begin{rem}
We consider that the optimization process is fast enough, so that the loads remain constant. Therefore we have $\Delta S_i=0$ for $i \notin \mathcal{V}_\text{src}\cup \mathcal{V}_\text{ren} \cup \mathcal{V}_\text{tf1} \cup \mathcal{V}_\text{tf2}$, where $\mathcal{V}_\text{tf1}, \mathcal{V}_\text{tf2}$ denote the set of nodes on the primary and secondary sides of the transformer respectively, and we do not need to measure or estimate the values of the loads. This is very relevant, since it allows to bypass the lack of load measurements. To achieve this fast optimization process, we can use an early stopping optimization or the recent work on projected gradient online optimization methods for the OPF problem in \cite{anese2016optimal, hauswirth2017online}. Moreover, this fast optimization allows to avoid outdated setpoints and thus suboptimal solutions as remarked in \cite{anese2016optimal}.
\end{rem}

Furthermore, if we consider a rectangular representation of the variables we can express \eqref{eq:PFaprox} linearly on the decision variables, $\Delta a_{\text{tap}}$ and the real and imaginary parts of $\Delta S, \Delta V$:

\begin{equation}\label{eq:PFapproxRect}
\arraycolsep=1pt
\begin{array}{rl}
\left[\begin{array}{c} 
\Delta P_\text{src} \\
\Delta P \\
\Delta Q_\text{src} \\ 
\Delta Q 
\end{array}\right] = & M
\left[\begin{array}{c} 
\Re\{\Delta V_\text{src}\} \\ 
\Re\{\Delta V\} \\
\Im\{\Delta V_\text{src}\} \\ 
\Im\{\Delta V\}
\end{array}\right]
\\[0.8cm]
\left[\begin{array}{c} 
\Re\{\Delta V_{\text{tf2}}\} \\
\Im\{\Delta V_{\text{tf2}}\} 
\end{array}\right] 
= &  a_\text{tap,prev}
\left[\begin{array}{c} 
\Re\{\Delta V_{\text{tf1}}\} \\
\Im\{\Delta V_{\text{tf1}}\} 
\end{array}\right] +
\Delta a_\text{tap}
\left[\begin{array}{c} 
\Re\{V_{\text{tf1},\text{prev}}\} \\
\Im\{V_{\text{tf1},\text{prev}}\} 
\end{array}\right] 
\end{array}
\end{equation}
where $M$ depends only on $V_\text{est}$ and $V_\text{src,prev}$, both known, and can be expressed as:
\rb{
\begin{equation*}
M = 
\left[\begin{array}{cc}
\Re\{A\}+\Re\{B\} & -\Im\{A\}+\Im\{B\} \\
\Im\{A\}+\Im\{B\} & \Re\{A\}-\Re\{B\}
\end{array}\right]
\end{equation*}
with
\begin{equation*}\arraycolsep=1pt
A = \text{diag}\left(\bar{Y}_{\text{isol}}
\left[\begin{array}{c}
\bar{V}_\text{src,prev} \\ \bar{V}_\text{est}
\end{array}\right]\right), 
B = \text{diag}\left(\left[\begin{array}{c}
V_\text{src,prev} \\ V_\text{est}
\end{array}\right] \right)\bar{Y}_{\text{isol}}
\end{equation*}
}

%Since $V_\text{prev}$ is not known, we will use $\bar{V}_\text{est}$ instead as operating point.

\subsection{Stochastic Voltage Limits}
Since we only have an estimation of the voltage previous to the optimization step \eqref{eq:Vprev}, the voltage after the optimization step $V$, $V_\text{rect} = [\Re \{V\}^T, \Im \{V\}^T]^T$, will be a prediction with a covariance:
\begin{equation}\label{eq:Vdistrib}
V_\text{rect} \sim \mathcal{N}(\Delta V_\text{rect} + V_\text{est,rect},\Sigma_\text{est,rect})
\end{equation} 

Since $V$ is then unknown, we cannot set a deterministic voltage limit constraint like \eqref{eq:opfvlim}. Therefore, we use a stochastic one instead:
\begin{equation}\label{eq:stochconstr}\begin{array}{c}
P(\abs{V}_{\min} \leq \abs{V_i} \leq \abs{V}_{\max}) \geq \beta \; \forall i 
\end{array}
\end{equation}
where $\beta$ is a desired threshold probability level, like $95\%$, \rb{which can be tuned to increase the confidence level of the constraint}. This constraint converts the OPF problem into a chance constrained optimization problem. We use the following theorem to reformulate \eqref{eq:stochconstr} as a function of our decision variables $\Delta V_\text{rect}$, our estimation $V_\text{est,rect}$, and $\Sigma_\text{est,rect}$ in \eqref{eq:Vprev}. We use $\Sigma_{\text{est,rect},\Re}$ and $\Sigma_{\text{est,rect},\Im}$ to denote the blocks corresponding to the covariance of the real and imaginary parts respectively.

\begin{thm}\label{thm:vlim}
For all $\beta \in (0,1)$, there exists $\alpha$ such that if the following constraints holds for all $i$:
\begin{equation}\arraycolsep=1pt\label{eq:detconstr}
\begin{array}{rl}
(\Re \{\Delta V_i\} + \Re \{V_{\text{est},i}\} \pm \alpha (\Sigma_{\text{est,rect},\Re})_{i,i}^\frac{1}{2})^2 & \\ 
+ (\Im \{\Delta V_i\} + \Im \{V_{\text{est},i}\} \pm \alpha (\Sigma_{\text{est,rect},\Im})_{i,i}^\frac{1}{2})^2 & 
\leq \abs{V}_{\max}^2  
\\[0.2cm]
(\Re \{\Delta V_i\} + \Re \{V_{\text{est},i}\} \pm \alpha (\Sigma_{\text{est,rect},\Re})_{i,i}^\frac{1}{2}) 
\frac{\Re \{V_{\text{est},i}\} }{\abs{V_{\text{est},i}}} & \\
+ (\Im \{\Delta V_i\} + \Im \{V_{\text{est},i}\} \pm \alpha (\Sigma_{\text{est,rect},\Im})_{i,i}^\frac{1}{2}) 
\frac{\Im \{V_{\text{est},i}\} }{\abs{V_{\text{est},i}}} & 
\geq \abs{V}_{\min} 
\end{array}
\end{equation}
then \eqref{eq:stochconstr} is satisfied. We use $\pm$ to denote all possible combinations to represent all constraints. 

Furthermore, $\alpha$ can be found using standard tables for  Gaussian distributions by choosing $\alpha$ such that 
\begin{equation}\label{eq:choosealpha}
P(\abs{\tilde{\omega}}\geq \alpha) \leq \frac{1-\beta}{4} \mbox{, for } \tilde{\omega}\sim \mathcal{N}(0,1)
\end{equation}

\end{thm}

\begin{proof}
We first split the constraint in \eqref{eq:stochconstr} into two independent conditions for $V_{\min}$ and $V_{\max}$:
\begin{equation*}\begin{array}{rl}
& P(\abs{V}_{\min} \leq \abs{V_i} \leq \abs{V}_{\max}) \\
= & P\big((\abs{V}_{\min} \leq \abs{V_i}) \cap (\abs{V_i} \leq \abs{V}_{\max})\big) \\
= & 1-P\big((\abs{V}_{\min} \geq \abs{V_i}) \cup (\abs{V_i} \geq \abs{V}_{\max})\big) \\
= & 1-P(\abs{V}_{\min} \geq \abs{V_i})-P(\abs{V_i} \geq \abs{V}_{\max})\\
& +P\big((\abs{V}_{\min} \geq \abs{V_i}) \cap (\abs{V_i} \geq \abs{V}_{\max})\big) \\
= & 1-P(\abs{V}_{\min} \geq \abs{V_i})-P(\abs{V_i} \geq \abs{V}_{\max}) \\
= & P(\abs{V}_{\min} \leq \abs{V_i})+P(\abs{V_i} \leq \abs{V}_{\max}) - 1
\end{array}
\end{equation*}
%so that 
%\begin{subequations}\label{eq:stochconstrsep2}
%\begin{align}
%P(\abs{V}_{\min} \leq \abs{V_i})
%% = P((\Re \{V_i\}^2+\Im \{V_i\}^2) \leq \abs{V}_{\max}^2)
%& \geq \frac{1+\beta}{2} \label{eq:stochconstrsepmin} \\
%\mbox{ and }
%P(\abs{V_i} \leq \abs{V}_{\max})
%% = P(\abs{V}_{\min}^2 \leq (\Re \{V_i\}^2 +\Im \{V_i\}^2)) 
%& \geq \frac{1+\beta}{2} \label{eq:stochconstrsepmax} \\
%\implies P(\abs{V}_{\min} \leq \abs{V_i} \leq \abs{V}_{\max}) & \geq \beta	
%\end{align}
%\end{subequations}
so that
\begin{equation}\label{eq:stochconstrsep}
\begin{array}{c}
P(\abs{V}_{\min} \leq \abs{V_i}) \geq \frac{1+\beta}{2} \mbox{ and } P(\abs{V_i} \leq \abs{V}_{\max}) \geq \frac{1+\beta}{2}
\\
\implies P(\abs{V}_{\min} \leq \abs{V_i} \leq \abs{V}_{\max}) \geq \beta	
\end{array}
\end{equation}
where for simplicity we have divided $1+\beta$ by $2$, but other options would also be possible. In rectangular coordinates, inequalities inside the probability in \eqref{eq:stochconstrsep} describe an annulus, which clearly is a non-convex region. Since we want to convert \eqref{eq:opfvlim} into a convex constraint, we consider the biggest convex region around the operating point \rb{$V_\text{est}$ and included in the annulus (see Fig. \ref{fig:region}):
\begin{equation}
\abs{V}_{\min} \leq \Re \{V_i\} \frac{\Re \{V_{\text{est},i}\}}{\abs{V_{\text{est},i}}} +\Im \{V_i\} \frac{\Im \{V_{\text{est},i}\}}{\abs{V_{\text{est},i}}} \Rightarrow 
\abs{V}_{\min} \leq \abs{V_i} 
\end{equation}
}
Then the lower bound in \eqref{eq:stochconstrsep} is replaced by the inequality representing the dark grey area in Fig. \ref{fig:region}:

\begin{equation}\label{eq:convexconstrmin}\arraycolsep=1pt\begin{array}{c}
P(\abs{V}_{\min} \leq \Re \{V_i\} \frac{\Re \{V_{\text{est},i}\}}{\abs{V_{\text{est},i}}} +\Im \{V_i\} \frac{\Im \{V_{\text{est},i}\}}{\abs{V_{\text{est},i}}}) \geq \frac{1+\beta}{2}
\end{array}
\end{equation}
%\pp{remark about what happens if SE gives absolute value}

\begin{figure}
\centering
\includegraphics[width=1.5in,height=1.1in]{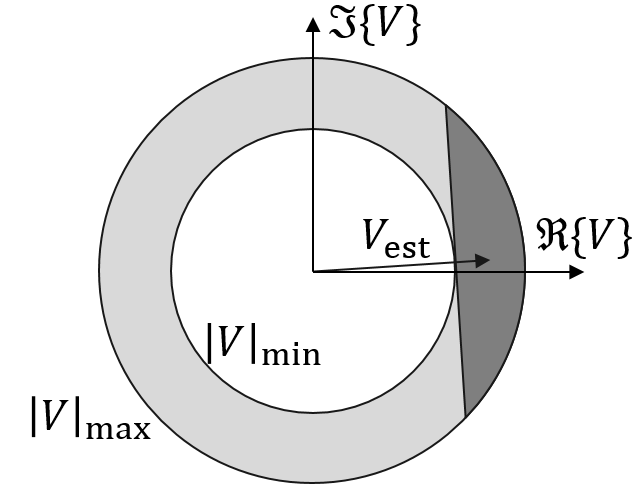}    % The printed column width is 8.4 cm.
\caption{Annulus non-convex region in light grey, with the biggest convex region around the operating point $V_\text{est}$ marked in dark grey.} 
\label{fig:region}
%\vspace{-0.6cm}
\end{figure}

Considering the estimation covariance in rectangular coordinates $\Sigma_\text{est,rect}$, using \eqref{eq:Vdistrib} we can expand $V_i$ \rb{as a function of the voltage deviations of our power flow formulation  \eqref{eq:PFapproxRect}}:
\begin{equation}\label{eq:Vexp}\arraycolsep=1pt
\begin{array}{rl}
\Re \{V_i\} & = \Re \{\Delta V_i\} + \Re \{V_{\text{est},i}\} + \omega_{i,\Re} \\
\Im \{V_i\} & = \Im \{\Delta V_i\} + \Im \{V_{\text{est},i}\} + \omega_{i,\Im}
\end{array}
\end{equation}
where $\omega_{i,\Re},\omega_{i,\Im}$, with $[\omega_{\Re}^T,\omega_{\Im}^T]^T \sim \mathcal{N}(0,\Sigma_\text{est,rect})$, denote the noises of the real and imaginary parts of the voltage estimation at node $i$: $V_{\text{prev},i}$, see \eqref{eq:Vprev}. \rb{Using the diagonal decomposition $\Sigma_\text{est,rect} = UD^2U^T$, we can rewrite these noises as different linear combinations of the same noise vector:
\begin{equation}\label{eq:noiseexp}\arraycolsep=1pt
\begin{array}{c}
\omega_{i,\Re} = U_{\Re,i\bullet}D \omega \sim \mathcal{N}(0,(\Sigma_{\text{est,rect},\Re})_{i,i}) \\
\omega_{i,\Im} = U_{\Im,i\bullet}D \omega \sim \mathcal{N}(0,(\Sigma_{\text{est,rect},\Im})_{i,i}) \\
\mbox{with }\omega \sim \mathcal{N}(0,I_\text{d})  
\mbox{ and }  U =[U_\Re^T \; U_\Im^T]^T
\end{array}
\end{equation}
}

Provided that conditions \eqref{eq:detconstr} are satisfied, \rb{we can use the voltage expressions \eqref{eq:Vexp} to express \eqref{eq:stochconstrsep} as a function of the noises $\omega_{i,\Re},\omega_{i,\Im}$:}
\begin{equation}\label{eq:Pmax}\arraycolsep=1pt\begin{array}{rrl}
& P&(\abs{V_i}^2 \leq \abs{V}_{\max}^2) \\
\stackrel{\eqref{eq:Vexp}}{=} & P&\big((\Re \{\Delta V_i\} + \Re \{V_{\text{est},i}\} + \omega_{i,\Re})^2  + 
\\ 
& & (\Im \{\Delta V_i\} + \Im \{V_{\text{est},i}\} + \omega_{i,\Im})^2 \leq \abs{V}_{\max}^2\big) 
\\
\geq & P&\Big(\big((\Re \{\Delta V_i\} + \Re \{V_{\text{est},i}\} + \omega_{i,\Re})^2 + 
\\ 
& & (\Im \{\Delta V_i\} + \Im \{V_{\text{est},i}\} + \omega_{i,\Im})^2 \leq \abs{V}_{\max}^2\big) \cap 
\\
& & (\abs{\omega_{i,\Re}} \leq \alpha(\Sigma_{\text{est,rect},\Re})_{i,i}^\frac{1}{2}) \cap (\abs{\omega_{i,\Im}} \leq \alpha(\Sigma_{\text{est,rect},\Im})_{i,i}^\frac{1}{2})\Big) 
\\
= & P&\Big(\big((\Re \{\Delta V_i\} + \Re \{V_{\text{est},i}\} + \omega_{i,\Re})^2 + 
\\ 
& & (\Im \{\Delta V_i\} + \Im \{V_{\text{est},i}\} + \omega_{i,\Im})^2 \leq \abs{V}_{\max}^2\big) \; \Big|
\\
& & (\abs{\omega_{i,\Re}} \leq \alpha(\Sigma_{\text{est,rect},\Re})_{i,i}^\frac{1}{2}) \cap (\abs{\omega_{i,\Im}} \leq \alpha(\Sigma_{\text{est,rect},\Im})_{i,i}^\frac{1}{2})\Big) 
\\
& \cdot & P\big((\abs{\omega_{i,\Re}} \leq \alpha(\Sigma_{\text{est,rect},\Re})_{i,i}^\frac{1}{2}) \cap (\abs{\omega_{i,\Im}} \leq \alpha(\Sigma_{\text{est,rect},\Im})_{i,i}^\frac{1}{2})\big)
\\
\stackrel{\eqref{eq:detconstr}}{=} & 1\cdot & P\big((\abs{\omega_{i,\Re}} \leq \alpha(\Sigma_{\text{est,rect},\Re})_{i,i}^\frac{1}{2}) \cap (\abs{\omega_{i,\Im}} \leq \alpha(\Sigma_{\text{est,rect},\Im})_{i,i}^\frac{1}{2})\big)
%\\
%\geq & 1\cdot & \frac{1+\beta}{2}
\end{array}
\end{equation}
and 
\begin{equation}\label{eq:Pmin}\arraycolsep=1pt\begin{array}{rrl}
& P&(\abs{V}_{\min} \leq \Re \{V_i\} \frac{\Re \{V_{\text{est},i}\}}{\abs{V_{\text{est},i}}} + \Im \{V_i\} \frac{\Im \{V_{\text{est},i}\}}{\abs{V_{\text{est},i}}}) 
\\
\stackrel{\eqref{eq:Vexp}}{=} & P&\big((\Re \{\Delta V_i\} + \Re \{V_{\text{est},i}\} + \omega_{i,\Re})\frac{\Re \{V_{\text{est},i}\}}{\abs{V_{\text{est},i}}}  + 
\\ 
& & (\Im \{\Delta V_i\} + \Im \{V_{\text{est},i}\} + \omega_{i,\Im})\frac{\Im \{V_{\text{est},i}\}}{\abs{V_{\text{est},i}}} \geq \abs{V}_{\min}\big) 
\\
\geq & P&\Big(\big((\Re \{\Delta V_i\} + \Re \{V_{\text{est},i}\} + \omega_{i,\Re})\frac{\Re \{V_{\text{est},i}\}}{\abs{V_{\text{est},i}}}  + 
\\ 
& & (\Im \{\Delta V_i\} + \Im \{V_{\text{est},i}\} + \omega_{i,\Im})\frac{\Im \{V_{\text{est},i}\}}{\abs{V_{\text{est},i}}} \geq \abs{V}_{\min}\big) \cap 
\\
& & (\abs{\omega_{i,\Re}} \leq \alpha(\Sigma_{\text{est,rect},\Re})_{i,i}^\frac{1}{2}) \cap (\abs{\omega_{i,\Im}} \leq \alpha(\Sigma_{\text{est,rect},\Im})_{i,i}^\frac{1}{2})\Big) 
\\
= & P&\Big(\big((\Re \{\Delta V_i\} + \Re \{V_{\text{est},i}\} + \omega_{i,\Re})\frac{\Re \{V_{\text{est},i}\}}{\abs{V_{\text{est},i}}}  + 
\\ 
& & (\Im \{\Delta V_i\} + \Im \{V_{\text{est},i}\} + \omega_{i,\Im})\frac{\Im \{V_{\text{est},i}\}}{\abs{V_{\text{est},i}}} \geq \abs{V}_{\min}\big) \; \Big| 
\\
& & (\abs{\omega_{i,\Re}} \leq \alpha(\Sigma_{\text{est,rect},\Re})_{i,i}^\frac{1}{2}) \cap (\abs{\omega_{i,\Im}} \leq \alpha(\Sigma_{\text{est,rect},\Im})_{i,i}^\frac{1}{2})\Big) 
\\
& \cdot & P\big((\abs{\omega_{i,\Re}} \leq \alpha(\Sigma_{\text{est,rect},\Re})_{i,i}^\frac{1}{2}) \cap (\abs{\omega_{i,\Im}} \leq \alpha(\Sigma_{\text{est,rect},\Im})_{i,i}^\frac{1}{2})\big)
\\
\stackrel{\eqref{eq:detconstr}}{=} & 1\cdot & P\big((\abs{\omega_{i,\Re}} \leq \alpha(\Sigma_{\text{est,rect},\Re})_{i,i}^\frac{1}{2}) \cap (\abs{\omega_{i,\Im}} \leq \alpha(\Sigma_{\text{est,rect},\Im})_{i,i}^\frac{1}{2})\big)
%\\
%\geq &1\cdot& \frac{1+\beta}{2}
\end{array}
\end{equation}

We can now use the distributions of $\omega_{i,\Re}$ and $\omega_{i,\Im}$ in \eqref{eq:noiseexp} to derive the way to choose $\alpha$ in \eqref{eq:choosealpha} \rb{to satisfy \eqref{eq:stochconstrsep}}. For a given threshold $\alpha$, since $\omega_{i,\Re}$ and $\omega_{i,\Im}$ are not independent, we can formulate the following inequality:
\begin{equation}\label{eq:omegaRI}\arraycolsep=1pt
\begin{array}{rl}
& P\big((\abs{\omega_{i,\Re}} \leq \alpha(\Sigma_{\text{est,rect},\Re})_{i,i}^\frac{1}{2}) \cap (\abs{\omega_{i,\Im}} \leq \alpha(\Sigma_{\text{est,rect},\Im})_{i,i}^\frac{1}{2})\big)
\\
= & 1\hspace{-0.1cm}-\hspace{-0.1cm}
P\big((\abs{\omega_{i,\Re}} \geq \alpha(\Sigma_{\text{est,rect},\Re})_{i,i}^\frac{1}{2}) \cup
(\abs{\omega_{i,\Im}} \geq \alpha(\Sigma_{\text{est,rect},\Im})_{i,i}^\frac{1}{2})\big)
\\
%= & 1\hspace{-0.1cm}-\hspace{-0.1cm} \Big(P(\abs{\omega_{i,\Re}} \geq \alpha(\Sigma_{\text{est,rect},\Re})_{i,i}^\frac{1}{2}) 
%\\
%& \hspace{0.7cm} + P(\abs{\omega_{i,\Im}} \geq \alpha(\Sigma_{\text{est,rect},\Im})_{i,i}^\frac{1}{2}) 
%\\
%& \hspace{0.7cm} - P\big(
%(\abs{\omega_{i,\Re}} \geq \alpha(\Sigma_{\text{est,rect},\Re})_{i,i}^\frac{1}{2}) \cap 
%\\ &
%\hspace{1.1cm} (\abs{\omega_{i,\Im}} \geq \alpha(\Sigma_{\text{est,rect},\Im})_{i,i}^\frac{1}{2})
%\big)\Big) 
%\\

= & 1 \hspace{-0.1cm}-\hspace{-0.1cm} P(\abs{\omega_{i,\Re}} \geq \alpha(\Sigma_{\text{est,rect},\Re})_{i,i}^\frac{1}{2}) 
\hspace{-0.1cm}-\hspace{-0.1cm} P(\abs{\omega_{i,\Im}} \geq \alpha(\Sigma_{\text{est,rect},\Im})_{i,i}^\frac{1}{2}) 
\\
& + P\big(
(\abs{\omega_{i,\Re}} \geq \alpha(\Sigma_{\text{est,rect},\Re})_{i,i}^\frac{1}{2}) \cap 
(\abs{\omega_{i,\Im}} \geq \alpha(\Sigma_{\text{est,rect},\Im})_{i,i}^\frac{1}{2})
\big)
\\

\geq & 1 \hspace{-0.1cm}-\hspace{-0.1cm} P(\abs{\omega_{i,\Re}} \geq \alpha(\Sigma_{\text{est,rect},\Re})_{i,i}^\frac{1}{2}) 
\hspace{-0.1cm}-\hspace{-0.1cm} P(\abs{\omega_{i,\Im}} \geq \alpha(\Sigma_{\text{est,rect},\Im})_{i,i}^\frac{1}{2})
\end{array}
\end{equation}
Using the Gaussian distribution tables to choose $\alpha$ such that 
\begin{equation}
P(\abs{\tilde{\omega}}\geq \alpha) \leq \frac{1-\beta}{4} \mbox{, for } \tilde{\omega}\sim \mathcal{N}(0,1)
\end{equation}
then we have
\begin{equation}\arraycolsep=1pt
\begin{array}{c}
P\big((\abs{\omega_{i,\Re}} \leq \alpha(\Sigma_{\text{est,rect},\Re})_{i,i}^\frac{1}{2}) \cap (\abs{\omega_{i,\Im}} \leq \alpha(\Sigma_{\text{est,rect},\Im})_{i,i}^\frac{1}{2})\big) \\ \geq \frac{1+\beta}{2}
\end{array}
\end{equation}

Using \eqref{eq:Pmax} and \eqref{eq:Pmin} we finally obtain:
\begin{equation}\arraycolsep=1pt
\begin{array}{l}
P(\abs{V_i}^2 \leq \abs{V}_{\max}^2) \geq \frac{1+\beta}{2}\\[0.1cm]
P(\abs{V}_{\min} \leq \Re \{V_i\} \frac{\Re \{V_{\text{est},i}\}}{\abs{V_{\text{est},i}}} + \Im \{V_i\} \frac{\Im \{V_{\text{est},i}\}}{\abs{V_{\text{est},i}}}) \geq \frac{1+\beta}{2}
\end{array}
\end{equation}
so that \eqref{eq:stochconstrsep} and \eqref{eq:convexconstrmin} are satisfied and so is \eqref{eq:stochconstr}.
\end{proof}

\begin{rem}
The convex approximation in \eqref{eq:convexconstrmin} may cause a loss of optimality in the OPF problem. However, it helps to avoid that the phase angle of the voltages deviates too much from the base voltages, which is also not desirable.
\end{rem}

\rb{The advantage of using our approach to deal with the chance constraints \eqref{eq:stochconstr} is that we separate the SE process from the OPF problem. Real-time sensors can be efficiently managed by the SE process \cite{picallo2017twostepSE}, and during the real-time operation of the network, we avoid sampling and/or computing convex approximations of the chance constraints using \rb{the conditional value at risk} as in \cite{dall2017chance}.}

\subsection{Final OPF}
Using Thm. \ref{thm:vlim} we can rewrite \eqref{eq:stochconstr} into a convex deterministic constraint, so that it can be integrated into our OPF problem:
\begin{defi} OPF with linear power flow approximation and stochastic voltage limits:
\begin{subequations}\label{eq:fopf}
\begin{align}\arraycolsep=1pt
&\begin{array}{l}
\mbox{Objective: }
\min \sum_{\phi} P_{\text{src},\phi} + Q_{\text{src},\phi}
\end{array}\label{eq:fopfobj}
\\
&\begin{array}{l}
\mbox{Constraints:}\\
\mbox{Power flow:}
\\ \;
\left[\begin{array}{c} 
\Delta P_\text{src} \\
\Delta P \\
\Delta Q_\text{src} \\ 
\Delta Q 
\end{array}\right] = M
\left[\begin{array}{c} 
\Re\{\Delta V_\text{src}\} \\ 
\Re\{\Delta V\} \\
\Im\{\Delta V_\text{src}\} \\ 
\Im\{\Delta V\}
\end{array}\right], 
\\[0.8cm] \;
\left[\begin{array}{c} 
\Re\{\Delta V_{\text{tf2}}\} \\
\Im\{\Delta V_{\text{tf2}}\} 
\end{array}\right] = 
\\[0.3cm] \;
a_\text{tap,prev}
\left[\begin{array}{c} 
\Re\{\Delta V_{\text{tf1}}\} \\
\Im\{\Delta V_{\text{tf1}}\} 
\end{array}\right] + 
\Delta a_\text{tap}
\left[\begin{array}{c} 
\Re\{V_{\text{tf1},\text{prev}}\} \\ 
\Im\{V_{\text{tf1},\text{prev}}\} 
\end{array}\right], 
\\[0.3cm] \;
0 = \Delta P_{\text{tf1}} + \Delta P_{\text{tf2}}, \; 0 = \Delta Q_{\text{tf1}} + \Delta Q_{\text{tf2}}, 
\\ \;
0 = \Delta P_i \mbox{ for } i \notin \mathcal{V}_\text{src} \cup \mathcal{V}_\text{ren} \cup \mathcal{V}_\text{tf1} \cup \mathcal{V}_\text{tf2}, 
\\ \;
0 = \Delta Q_i \mbox{ for } i \notin \mathcal{V}_\text{src} \cup \mathcal{V}_\text{ren} \cup \mathcal{V}_\text{tf1} \cup \mathcal{V}_\text{tf2}
\end{array}\label{eq:fopfpf}
\\
&\begin{array}{l}
\mbox{Tap changers: }
\\ \;
a_{\text{tap},\phi} \in [a_{\text{tap},\min},a_{\text{tap},\max}], \; \forall \phi \in \{1,2,3\}
\end{array}\label{eq:fopftap}
\\
&\begin{array}{l}
\mbox{Available energy:}
\\ \;
P_{\min,i} \leq \Delta P_i + P_{\text{prev},i} \leq P_{\max,i}, \; \forall i \in \mathcal{V}_\text{ren}
\\ \;
Q_{\min,i} \leq \Delta Q_i + Q_{\text{prev},i} \leq Q_{\max,i}, \; \forall i \in \mathcal{V}_\text{ren} 
\\ \;
\abs{\Delta S_i+ S_{\text{prev},i}} \leq \abs{S}_{\max,i}, \; \forall i \in \mathcal{V}_\text{ren}
\end{array}\label{eq:fopfenergy}
\\
&\begin{array}{l}
\mbox{Voltage limits: } \forall i \in \{1,\ldots,N\}
\\ \;
(\Re \{\Delta V_i\} + \Re \{V_{\text{est},i}\} \pm \alpha (\Sigma_{\text{est,rect},\Re})_{i,i}^\frac{1}{2})^2 
\\ \;
+(\Im \{\Delta V_i\} + \Im \{V_{\text{est},i}\} 
\pm \alpha (\Sigma_{\text{est,rect},\Im})_{i,i}^\frac{1}{2})^2 \leq \abs{V}_{\max}^2 \\ \;
\mbox{and} \\ \;
(\Re \{\Delta V_i\} + \Re \{V_{\text{est},i}\} 
\pm \alpha  (\Sigma_{\text{est,rect},\Re})_{i,i}^\frac{1}{2}) \frac{\Re V_{\text{est},i}}{\abs{V_{\text{est},i}}} 
\\ \;
+(\Im \{\Delta V_i\} + \Im \{V_{\text{est},i}\} 
\pm \alpha (\Sigma_{\text{est,rect},\Im})_{i,i}^\frac{1}{2}) \frac{\Im V_{\text{est},i}}{\abs{V_{\text{est},i}}} \geq \abs{V}_{\min}
\end{array}\label{eq:fopfvlim}
\end{align}
\end{subequations}
\end{defi}
where now the variables to control are $\Delta a_\text{tap}$, $\{(\Delta P_i, \Delta Q_i) \mid i \in \mathcal{V}_\text{ren}\}$ and $\Delta V_\text{src}$, while $\Delta P_\text{src}$, $\Delta Q_\text{src}$ and $\Delta V$ are determined by the constraints. All previous values before the optimization step, $a_\text{tap,prev}, P_{\text{prev},i}, Q_{\text{prev},i}, S_{\text{prev},i}$, are stored from the control step before the current one, and $V_\text{est}$ and $\Sigma_\text{est,rect}$ are given by the SE.
\begin{rem}
Note that the OPF problem in \eqref{eq:fopf} is convex, and thus can be solved optimally and efficiently using state-of-the-art convex-optimization algorithms \cite{boyd2004convex}.
\end{rem}

\section{Case Study}\label{sec:sim}
A simulation of 24 hours with 15 min intervals is run on a test case to test the effectiveness of the methodology. Here we describe the settings of the test case and analyze the results. The algorithms are coded in Python and run on an Intel Core i7-5600U CPU at 2.60GHz with 16GB of RAM.
\subsection{Settings}\label{subsec:simDef}
\begin{itemize}
\item System: We use the 123-bus test feeder available online \cite{kersting1991radial, testfeeder}, see Fig. \ref{fig:123bus}. 
\item Measurements for the SE (see Fig. \ref{fig:123bus}): As in \cite{picallo2017twostepSE}, voltage measurements (red circle for phasor, red square for magnitude only) are placed at buses $95$ and $83$, current measurements (blue dashed circle for phasor, blue dashed square for magnitude only) at buses $65$ and $48$, and branch current phasor measurements (blue dashed arrow) at branch $150$ (after the regulator) $\to 149$. 
\item Load Profiles: As in \cite{picallo2017twostepSE} for the SE, the load profiles are built by aggregating a several households profiles.
\item Distributed Generation: Solar energy is introduced in the three phases of nodes $49$ and $65$, and wind energy in nodes $76$ and $30$, see Fig \ref{fig:123bus} (a yellow rhombus for solar, a grey parallelogram for wind). The profiles can be seen in Fig. \ref{fig:energyused}. They are simulated using a solar irradiation profile and a wind speed profile from \cite{solarprofile, windprofile}. We use these profiles to determine for each time step the apparent power limit $\abs{S}_{\max,i}$ for $i \in \mathcal{V}_\text{ren}$ in \eqref{eq:fopfenergy}. 
\item Tap changers: We control the transformer located in the branch $160\to 67$, see Fig. \ref{fig:123bus}, and set the tap changers limits at $a_{\text{tap},\min}=0.9,a_{\text{tap},\max}=1.1$, as in \cite{robbins2016optimal}.
\item Voltage limits: The common values $\abs{V}_{\max} = 1.05 \mathrm{ p.u.}$ and $\abs{V}_{\min} = 0.95 \mathrm{ p.u.}$ are used for the voltage magnitude constraints \cite{robbins2016optimal, anese2016optimal}.
\item Other values: The probability limit for the stochastic constraints \eqref{eq:stochconstr} is chosen to be the standard value $\beta = 95\%$, resulting in $\alpha \approx 2.5$. 
\end{itemize}

\begin{figure}
\centering
\includegraphics[width=9cm,height=7.3cm]{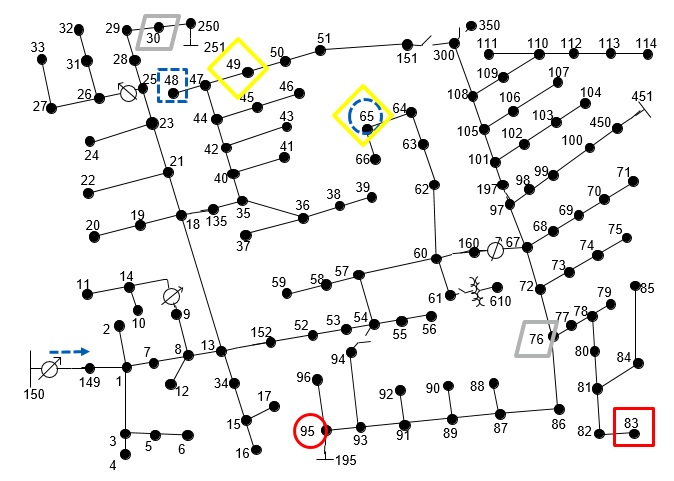}    % The printed column width is 8.4 cm.
\caption{123-bus test feeder with location of measurements and distributed generation marked. The network image has been taken from \cite{testfeeder}.} 
\label{fig:123bus}
%\vspace{-0.6cm}
\end{figure}

%\begin{figure}
%\centering
%\includegraphics[width=2.66in,height=2in]{fig/figloads}    % The printed column width is 8.4 cm.
%\caption{Load profiles are built by aggregating households profiles from the dataset in \cite{pedersen2015disc}.} 
%\label{fig:figloads}
%%\vspace{-0.6cm}
%\end{figure}

\subsection{Results}\label{subsec:simRes}
To prove the effectiveness of our approach, we compare the resulting voltage magnitudes when controlling the transformer and the introduced energy in two cases: in case 1, we take the uncertainty into account, using the covariance to ensure the voltage limit constraints; while in case 2, we are using the voltage estimates as if they were the true values, without taking into account the covariance. It can be observed in Fig. \ref{fig:Vboth} that case 1, using the covariance, performs much better than case 2 in controlling the voltage magnitudes within their limits.

From the case 1 using the covariance, we can also observe in Fig. \ref{fig:tap} how the tap changers remain within the limits, and their value changes to optimize the operation of the grid. Furthermore, we can also observe in Fig. \ref{fig:energyused} how the whole energy available is not always fully used. This happens because otherwise some voltage constraints would be violated. Precisely, in Fig. \ref{fig:Vboth} it can be observed that the instants at which some of the available energy is not used, coincide with some voltage magnitudes being at the limit.

Moreover, note that the load profiles introduced in the SE are not necessarily normally distributed, and thus the SE estimate may have a non-Gaussian probability distribution. However, our approach in case 2 still succeeds in keeping the voltages within limits.

%In Fig. \ref{fig:energysubs}

\begin{figure}
\centering
\includegraphics[width=9cm,height=5.5cm]{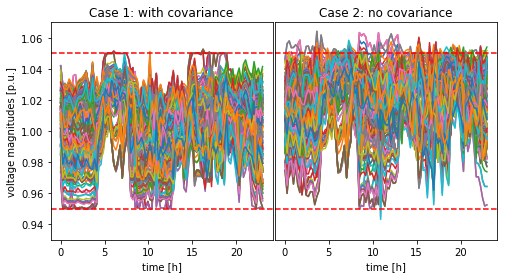}    % The printed column width is 8.4 cm.
\caption{Voltage magnitude $\abs{V}$ profiles for all nodes along the day for case 1 and 2: with and without taking into account the covariance of the SE estimate for the OPF respectively. The red dashed lines represent the limits.} 
\label{fig:Vboth}
%\vspace{-0.6cm}
\end{figure}

\begin{figure}
\centering
\includegraphics[width=2.66in,height=2in]{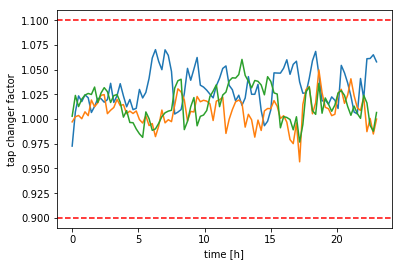}    % The printed column width is 8.4 cm.
\caption{Values of the tap changers at the transformers for the 3 phases. The red dashed lines represent the limits.} 
\label{fig:tap}
%\vspace{-0.6cm}
\end{figure}

\begin{figure}
\centering
\includegraphics[width=8.5cm,height=8.5cm]{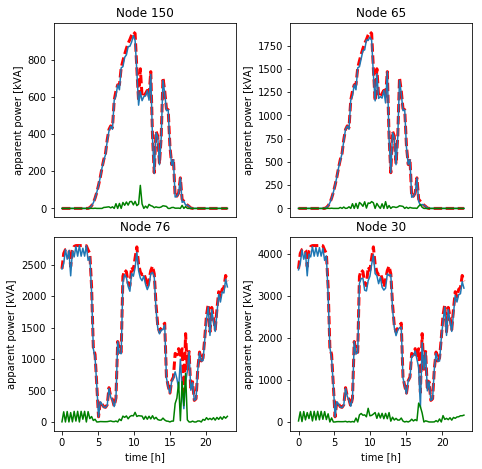}    % The printed column width is 8.4 cm.
\caption{Profile of available energy at every node (thick red dashed line), the actual energy taken (blue line), and the difference (green).} 
\label{fig:energyused}
%\vspace{-0.6cm}
\end{figure}

%\begin{figure}
%\centering
%\includegraphics[width=8.5cm,height=8.5cm]{fig/energysubs}    % The printed column width is 8.4 cm.
%\caption{Profile of active and reactive energy taken from the substation (blue line for case 1, red dashed line for case 2).} 
%\label{fig:energyused}
%%\vspace{-0.6cm}
%\end{figure}

%\begin{figure}
%\captionsetup[subfigure]{aboveskip=-0.5pt,belowskip=-0.5pt}
%\begin{center}
%\begin{subfigure}{.143\textwidth}
%\includegraphics[width=\linewidth]{fig/errorspermethodNLa} 
%\caption{}\label{fig:errorsmethoda}
%\end{subfigure}%
%\begin{subfigure}{.26\textwidth}
%\includegraphics[width=\linewidth]{fig/errorspermethodNLb} 
%\caption{}\label{fig:errorsmethodb}
%\end{subfigure}
%\vspace{-0.2cm}\caption{Box-plots of MAPE errors for the different solutions. Red line indicate the median, red square the mean.} 
%\includegraphics[width=7.4cm,height=5cm]{fig/exectimeNL}    % The printed column width is 8.4 cm.
%\vspace{-0.2cm}\caption{Box-plots of execution times for the different solutions. Red line indicate the median, red square the mean.} 
%\label{fig:exectime}
%\end{center}
%\vspace{-0.7cm}
%\end{figure}

\section{Conclusions}\label{sec:conc}
In this work, we have presented a methodology to combine the state estimation (SE) with the optimal power flow (OPF) problem for a distribution grid where only a few measurements are available. The lack of sensors produces uncertain voltage state estimates, and therefore we have adapted the standard OPF to include stochastic constraints for the voltage magnitude levels. Moreover, we use a linear approximation of the power flow to convexify the problem and bypass the lack of load measurements using delta increments. We also transform the stochastic constraints to make them tractable for our problem, by using the probability distribution of the state estimation. Finally, we prove through a case study, that the proposed methodology succeeds in controlling a large distribution grid with some controllable elements, like transformers and distributed generation, while respecting the voltage constraints. 

Future work could include adding controllable elements with dynamics, like batteries, or flexible loads, like electrical vehicles. \rb{Additionally, we could also consider the discrete nature of the tap changer values in the transformer, as well as other probability distributions for the voltage SE.}

\bibliographystyle{ieeetr}
\bibliography{ifacconf}

%\bibliographystyle{IEEEtran}

\end{document}